\documentclass[12pt]{article}

\newcommand{\blind}{1}

\usepackage[margin=1.2in]{geometry}

\usepackage[utf8]{inputenc}
\usepackage{natbib}
\usepackage{amssymb,amsbsy,amsfonts,amsmath,xspace,amsthm}
\usepackage{mathrsfs}
\usepackage{graphicx}

\usepackage[colorlinks,citecolor=blue]{hyperref}

\usepackage{float}
\usepackage{subfigure}
\usepackage{amsthm}
\usepackage{dsfont,color,bm,mathtools,enumitem}
\mathtoolsset{showonlyrefs}
\newtheorem{thm}{Theorem}

\theoremstyle{definition}
\newtheorem{defn}{Definition}

\usepackage{stmaryrd} 

\def\mxi{\boldsymbol{\xi}}

\def\bmu{\bm{u}}
\def\mv{\bm{v}}

\def\mA{\bm{A}}
\def\mB{\bm{B}}

\def\mH{\bm{H}}
\def\mI{\bm{I}}

\def\mI{\bm{I}}
\def\mM{\bm{M}}

\def\mP{\bm{P}}

\def\mU{\bm{U}}
\def\mV{\bm{V}}

\def\mX{\bm{X}}

\def\mZ{\bm{Z}}

\def\mZ{\bm Z}
\def\mA{\bm A}
\def\mB{\bm B}

\def\mI{\bm I}

\def\mX{\bm X}

\def\mM{\bm M}
\def\mSigma{\bm \Sigma}

\def\tZ{\mathcal{Z}}
\def\tA{\mathcal{A}}

\def\tC{\mathcal{C}}

\def\tM{\mathcal{M}}
\def\tN{\mathcal{N}}
\def\tO{\mathcal{O}}

\def\tS{\mathcal{S}}

\def\tX{\mathcal{X}}

\def\tZ{\mathcal{Z}}

\def\bbR{\mathbb{R}}

\newcommand{\FnormSize}[2]{#1\lVert#2#1\rVert_F}

\newcommand{\indep}{\perp \!\!\! \perp}

\def\bSig\mathbf{\Sigma}

\usepackage{hyperref}
\hypersetup{
    colorlinks=true,
    linkcolor=blue,
    filecolor=magenta,      
    urlcolor=cyan,
    citecolor=blue
}

\usepackage{dsfont,listings}

\usepackage{xcolor}
\allowdisplaybreaks

\usepackage{setspace}
\begin{document}

\if1\blind
{   \date{}
  \title{\bf Sufficient dimension reduction for feature matrices}
\author{Chanwoo Lee\\
Department of Statistics, University of Wisconsin-Madison\\ chanwoo.lee@wisc.edu}

    \maketitle
} \fi

\if0\blind
{
 \date{}
  \title{\bf Statistical and computational rates in high rank tensor estimation}
\author{}
\maketitle
} \fi

\begin{abstract}
We address the problem of sufficient dimension reduction for feature matrices, which arises often in sensor network localization, brain neuroimaging, and electroencephalography analysis. In general, feature matrices have both row- and column-wise interpretations and contain structural information that can be lost with naive vectorization approaches. To address this, we propose a method called principal support matrix machine (PSMM) for the matrix sufficient dimension reduction.
The PSMM converts the sufficient dimension reduction problem into a series of classification problems by dividing the response variables into slices. It effectively utilizes the matrix structure by finding hyperplanes with rank-1 normal matrix that optimally separate the sliced responses. Additionally, we extend our approach to the higher-order tensor case.
Our numerical analysis demonstrates that the PSMM outperforms existing methods and has strong interpretability in real data applications.
\end{abstract}
\noindent%

\section{ Introduction}

Matrix-valued datasets are ubiquitous in modern data science applications. For example, electroencephalography (EEG) data collects data from 122 subjects in two groups: an alcoholic and a control group. 
Each subject was exposed to a stimulus, and the scalp of the subjects was fitted with 64 electrodes that recorded voltage values for 256 time points. As a result, each sampling unit consists of a 256 $\times$ 64 matrix with a group label. Understanding the relationship between alcoholism and the voltage patterns across time and channels is of scientific interest.
Another example includes the MRN-114 human brain connectivity data. The data consists of $114$ subjects along with their Full Scale Intelligence Quotient (FSIQ) score. For each subject, a binary connectivity matrix among 68 brain regions is collected based on the Desikan atlas~\citep{desikan2006automated}. Learning from this matrix-valued dataset provides interesting insights about the association between intelligence and brain connectivity.

Let $\mX\in\mathbb{R}^{d_1\times d_2}$ be a matrix predictor and $Y\in\mathbb{R}$ be a response variable. We are interested in reducing the dimension of the matrix $\mX$ without losing the regression relation between $\mX$ and $Y$. One naive approach is to vectorize the feature matrix and apply classical dimension reduction methods to estimate a  matrix $\mM\in\mathbb{R}^{d_1d_2\times r}$ with $r<d_1d_2$ such that
\begin{align} \label{eq:sdrv}
Y\indep \text{vec}(\mX) | \mM^T\text{vec}(\mX),
\end{align}
where $\text{vec}\colon\mathbb{R}^{d_1\times d_2}\rightarrow\mathbb{R}^{d_1d_2}$ is a linear transformation that converts the matrix into a column vector. This classical sufficient dimension reduction problem has received much attention, and many methods have been proposed and studied \citep{li1991sliced,duan1991slicing,cook1991discussion,cook2005sufficient,Li2005ContourRA,li2007directional,artemiou2016sufficient}.

However, matrices are not simply vectors with additional indices; instead, they possess structural information which a simple vectorization approach fails to exploit. For instance, in EEG data, each row and column of matrix $\mX$ corresponds to a specific electrode channel and time point. If we were to vectorize this feature matrix, this meaningful interpretation would be lost. Similarly, in MRN-114 human brain connectivity data, the brain network is naturally represented as symmetric adjacency matrices $\mX$, where the values signify the presence or absence of fiber connections. Converting these matrices into vectors would cause the loss of symmetry, and the information contained within it could not be well utilized. Additionally, if the feature matrix is transformed into a vector by stacking its columns or rows, the resulting vector would have a very high dimensionality. This could lead to the curse of dimensionality. By keeping $\mX$ as a matrix, the number of parameters is reduced, and the accuracy of estimation can be improved.

To leverage the matrix structure, we consider the following objective to find two matrices $\mU\in\mathbb{R}^{d_1\times r_1}$ $(r_1< d_1)$ and $\mV\in\mathbb{R}^{d_2\times d_2}$ $(r_2< d_2)$ such that
\begin{align}\label{eq:main}
    Y\indep \mX|\mU^T\mX\mV.
\end{align}
Then, we keep all information of a feature matrix $\mX\in\mathbb{R}^{d_1\times d_2}$ to predict the response $Y$ reducing its dimension to $\mU^T\mX\mV\in\mathbb{R}^{r_1\times r_2}$. At the same time, we still preserve matrix structural information, including row- and column-wise interpretation.

Notice that the identifiable parameters are the subspaces spanned by the matrices $\mU$ and $\mV$, i.e., $\text{span}(\mU)$ and $\text{span}(\mV)$ respectively. In fact, \eqref{eq:main} is equivalent to
\begin{align}\label{eq:main2}
Y\indep \mX|(\mU\mA)^T\mX(\mV\mB),
\end{align}
for any nonsingular matrices $\mA\in\mathbb{R}^{r_1\times r_1}$ and $\mB\in\mathbb{R}^{r_2\times r_2}$. Thus, our goal is to estimate the column spaces of $\mU$ and $\mV$ rather than $\mU$ and $\mV$ themselves. We call the subspace induced by $\mU$ (and $\mV$) satisfying \eqref{eq:main}, row (and column) dimension reduction subspace. This notion was first proposed in \cite{li2010dimension} as left and right dimension-folding subspace.
We define the central row and column subspace similar to the central subspace defined in the vector case. In the vector case, it is well-known that the intersection of two dimension reduction subspaces is itself a dimension reduction subspace~\citep{chiaromonte2002sufficient,yin2008successive}. Similarly, in the matrix case, the intersection of row dimension reduction subspaces for $Y|\mX$ is again a row dimension reduction subspace under mild conditions~\citep{li2010dimension}. The same argument holds for the column dimension reduction subspace.
Thus, we define the central row and column subspaces in the following way.
\begin{defn}[Central subspace for  matrix]\label{def:central}
Define the central row and column subspaces as
\begin{align}
S^r_{Y|\mX} &= \underset{{\mU\colon Y\indep \mX|\mU^T\mX\mV}}{\cap} \text{span}(\mU)\quad\text{and}\quad  S^c_{Y|\mX} &= \underset{{\mV\colon Y\indep \mX|\mU^T\mX\mV}}{\cap} \text{span}(\mV).
\end{align}
The subspace $S_{Y|\mX} = S^r_{Y|\mX} \times S^c_{Y|\mX}$ is called the central subspace for matrices.
\end{defn}
We can also rewrite the central subspace as
\begin{align}
S_{Y|\mX} &= \underset{{(\mU,\mV)\colon Y\indep \mX|\mU^T\mX\mV}}{\cap} \text{span}(\mU)\times \text{span}(\mV).
\end{align}
The central dimension-folding subspace in \cite{li2010dimension,ding2015tensor} is equivalent to our central subspace with the Cartesian product $\times$ replaced by the Kronecker product $\otimes$. Here, we adopt the Cartesian product for a cleaner exposition and easier generalization to the higher-order tensor case. The extension to the central subspace for a higher-order tensor is presented in Section~\ref{sec:extend}.

In this paper, our goal is to estimate the central subspace for matrices and we propose the principal support matrix machine (PSMM).

\subsection{Related works and our contribution}
Our research is closely connected to, yet also has distinct differences from several existing lines of works. In this section, we review related literature and remark our contribution.
\paragraph{Principal support vector machine}  Our PSMM is closely related to the principal support vector machine (PSVM) proposed in \cite{li2011principal, artemiou2016sufficient}. The PSVM considers sufficient dimension reduction in \eqref{eq:sdrv} for the vector case. The main idea of the PSVM is to divide feature vectors into several slices based on the value of the responses and obtain hyperplanes that optimally separate these slices using support vector machine. The aggregation of these hyperplanes by applying principal component analysis provides a consistent estimator of the central subspace for vectors. However, the PSVM only allows for vector-valued predictors, and vectorizing matrix-valued predictors loses structural information and leads to high dimensionality. By contrast, our PSMM provides an efficient sufficient dimension reduction method for matrix predictors and successfully preserves structural information. We observe a clear improvement of matrix-based methods in our numerical studies.

\paragraph{Support matrix machine}
Applications of the support matrix machine (SMM) have shown great success in image classification, visual recognition, and EEG data analysis~\citep{pirsiavash2009bilinear,luo2015support}. The SMM is proposed and developed for the classification problem with matrix predictors. It considers the following formulation, extending support vector machines to the matrix case:
\begin{align}\label{eq:SMM}
    \min_{\mB,b} \|\mB\|_F^2 + \frac{\lambda}{n}\sum_{i=1}^n\left\{1-Y_i(\langle \mB,\mX_i-\bar \mX\rangle + b)\right\}_+,
\end{align}
where  $\{x\}_+ = \max(x,0)$, $\mB\in\mathbb{R}^{d_1\times d_2}$ is a coefficient matrix with low-rankness, $b\in\mathbb{R}$ is an intercept, and $\lambda$ is a positive penalty parameter. By imposing the low-rank structure on the coefficient normal matrix, the SMM utilizes the structural information of the feature matrix. We adapt this idea and apply it to the sufficient dimension reduction context. We turn the sufficient dimension reduction problem into a series of classification problems, where we use SMM techniques with modification. We demonstrate that the PSMM effectively estimates the central subspace for matrices, leveraging the benefits of SMM methods such as the consideration of matrix structure, robustness against outliers, and efficiency in high-dimensional settings.

\paragraph{Dimension folding} \cite{li2010dimension,ding2015tensor} introduced the central dimension folding subspace, which is equivalent to the central subspace for matrices as defined in Definition~\ref{def:central}. They also proposed methods for estimating the central subspace generalizing existing inverse regression-based methods to the matrix and higher-order tensor case. Such methods include sliced inverse regression (SIR)~\citep{li1991sliced}, the sliced average variance estimator (SAVE)~\citep{cook1991discussion}, and directional regression (DR)~\citep{li2007directional}. However, as pointed out in \cite{li2011principal}, such methods tend to downweight the slice means near the center of data due to its shorter length. This characteristics often makes these methods inaccurate since it is known that a regression surface is well estimated at the center of the data. In contrast, the PSMM finds coefficient normal matrices that optimally separate data points depending on the sliced responses. This approach allows the appropriate use of data points near the center. We demonstrate that the PSMM indeed improves accuracy over inverse regression-based methods in Section~\ref{sec:sim}.

\subsection{Notation and organization}
 We use the shorthand $[n]$ to denote $\{1,\ldots,n\}$ for $n\in\mathbb{N}_{+}$. For any two matrices $\mA,\mB\in\mathbb{R}^{d_1\times d_2}$, the inner product of two matrices is defined as $\langle\mA, \mB\rangle = \sum_{(i,j)\in[d_1]\times[d_2]}\mA_{ij}\mB_{ij}$. 
For a matrix $\mA\in\mathbb{R}^{d_1\times d_2}$, we use $\lambda_i(\mA)$ to denote i-th largest eigenvalue of $\mA$ and $\FnormSize{}{\mA}= \sqrt{\sum_{(i_1,i_2)\in[d_1\times[d_2]}\mA_{i_1i_2}^2}$ to denote its Frobenius norm. 
Let $\tA\in\bbR^{d_1\times \cdots \times d_K}$ be an order-$K$ $(d_1,\ldots,d_K)$-dimensional tensor and $\tA_{i_1,\ldots,i_K}$  the tensor entry indexed by $(i_1,\ldots,i_K)\in[d_1]\times\cdots\times[d_K]$.  We define Frobenius norm of tensor $\tA$ as $\FnormSize{}{\tA} = \sqrt{\sum_{(i_1,\ldots,i_K)\in[d_1]\times \cdots\times [d_K]}\tA_{i_1,\ldots,i_K}^2}$. The multilinear multiplication of a tensor $\tC\in\mathbb{R}^{r_1,\ldots,r_K}$ by matrices $\mU_k\in\mathbb{R}^{d_k\times r_k}$, $k\in[K]$ is defined as 
\begin{align}
    (\tC\times_1\mU_1\times\cdots\times_K\mU_K)_{i_1,\ldots,i_K} = \sum_{j_1 = 1}^{r_1}\cdots\sum_{j_K = 1}^{r_K}\tC_{j_1,\ldots,j_d}(\mU_1)_{i_1j_1}\cdots(\mU_K)_{i_Kj_K},
\end{align}
which results in an order-$K$ $(d_1,\ldots,d_K)$-dimensional tensor.
 
The rest of the paper is organized as follows. Section~\ref{sec:psmm} introduces an objective function of the PSMM at the population level and constructs the unbiasedness of the estimator. We then present the estimation procedure for the matrix sufficient dimension reduction at the sample level in Section~\ref{sec:psmmalg}. In Section~\ref{sec:extend}, we extend all the results of the matrix case to the higher-order tensor case. Synthetic and real data analyses are presented in Section~\ref{sec:sim}. We conclude the paper with a discussion in Section~\ref{sec:disc}.

\section{Principal support matrix machine at the population level}\label{sec:psmm}
In this section, we present an objective function of the PSMM at the population level and provide intuition of the PSMM for the  matrix sufficient dimension reduction.

We first consider the binary classification problem for feature matrices. For now, we assume the response $Y$ to be binary values of -1 or 1. We introduce the rank-1 support matrix machine (SMM) with the samples $\{(Y_i,\mX_i)\}_{i=1}^n$. Plugging in the rank-1 coefficient matrix $\mB = \bmu\mv^T$ into the SMM in Equation \eqref{eq:SMM} yields the rank-1 SMM:
\begin{align}\label{eq:1smm}
    \min_{(\bmu,\mv,t)\in\mathbb{R}^{d_1}\times\mathbb{R}^{d_2}\times\mathbb{R}}\quad (\bmu^T\bmu)(\mv^T\mv) + \frac{\lambda}{n}\sum_{i=1}^n\left\{1-Y_i[\bmu^T(\mX_i-\bar\mX)\mv-t]\right\}_+.
\end{align}
This SMM objective function is an extension of the SVM, and the rank-1 constraint helps to utilize structural information  of matrix predictors. The solution of Equation \eqref{eq:1smm}, denoted as $(\bmu^*,\mv^*,t^*)$, defines the optimal hyperplane $\{\mX\colon (\bmu^*)^T\mX\mv^* = t^*\}$ that separates the two spaces $\{\mX_i\colon Y_i = 1\}$ and $\{\mX_i\colon Y_i = -1\}$.

Now we consider the matrix sufficient dimension reduction problem, where the response $Y$ can be continuous variable.
Let $\Omega_Y$ be the support of $Y$. Let $A_1$ and $A_2$ be arbitrary disjoint subsets of $\Omega_Y$. Define  the discrete random variable $\tilde Y$ such that
\begin{align}
    \tilde Y = \mathds{1}\{Y\in A_1)-\mathds{1}\{Y\in A_2\}.
\end{align}
We propose the following objective function at the population level for the matrix sufficient dimension reduction:
\begin{align}\label{eq:obj}
    L(\bmu,\mv,t) &= \text{Var}(\bmu^T\mX\mv)+\lambda \mathbb{E}\left\{1-\tilde Y(\bmu^T(\mX-\mathbb{E}(\mX))\mv-t)\right\}_+.
\end{align}
Compared to the rank-1 SMM in Equation \eqref{eq:1smm}, we consider the variance factor of $\mX$. If $\text{Var}\left[\text{vec}(\mX)\right] = I_{d_1+d_2}$, the objective function in Equation \eqref{eq:obj} reduces to the population version of Equation \eqref{eq:1smm}. This variance consideration establishes the unbiasedness of an  estimator which minimizes \eqref{eq:obj} for the central subspace, as suggested in the following theorem.
\begin{thm}\label{thm:unbiased}
Suppose that $\mathbb{E}(\mX|\mU^T\mX\mV)$ is a bilinear function of $\mU^T\mX\mV$, where  $\mU$ and $\mV$ are matrices as defined in \eqref{eq:main}. If $(\bmu^*,\mv^*,t^*)$ minimizes the objective function \eqref{eq:obj} among all $(\bmu,\mv,t)\in\mathbb{R}^{d_1}\times\mathbb{R}^{d_2}\times \mathbb{R}$, then $(\bmu^*,\mv^*)\in S_{Y|\mX}.$
\end{thm}

\begin{proof}
Without loss of generality, assume that $\mathbb{E}(\mX) = 0_{d_1\times d_2}$. Notice that 
\begin{align}
    \mathbb{E}\left[1-\tilde Y(\bmu^T\mX\mv-t)\right]_+ = \mathbb{E}\left[\mathbb{E}\left[(1-\tilde Y(\bmu^T\mX\mv-t))_+|Y,\mU^T\mX\mV\right]\right]
\end{align}
By Jensen's inequality, we have 
\begin{align}
    \mathbb{E}&\left[(1-\tilde Y(\bmu^T\mX\mv-t))^+|Y,\mU^T\mX\mV\right]\\&\geq \mathbb{E}\left[(1-\tilde Y(\bmu^T\mX\mv-t))|Y,\mU^T\mX\mV\right]_+
    \\&=\left[1-\tilde Y\left[\mathbb{E}((\bmu^T\mX\mv)|Y,\mU^T\mX\mV)-t\right]\right]_+,
\end{align}
where the first equality follows from $Y\indep \mX|\mU^T\mX\mV$. 
Therefore, we have
\begin{align}\label{eq:mean}
    \mathbb{E}\left[1-\tilde Y(\bmu^T\mX\mv-t)\right]_+&\geq\mathbb{E}\left[1-\tilde Y\left[\mathbb{E}((\bmu^T\mX\mv)|Y,\mU^T\mX\mV)-t\right]\right]_+
    \nonumber\\&= \mathbb{E}\left[1-\tilde Y\left[(\mU\eta_r)^T\mX(\mV\eta_c)-t\right]\right]_+,
\end{align}
 for some $\eta_r\in\mathbb{R}^{r_1},\eta_c\in\mathbb{R}^{r_2}$. We can always find such $\eta_r,\eta_c$ because $\mathbb{E}(\mX|\mU^T\mX\mV)$ is a bilinear function of $\mU^T\mX\mV$.
Also, notice that 
\begin{align}\label{eq:var}
    \text{Var}&(\bmu^T\mX\mv)\nonumber \\&= \text{Var}\left[\mathbb{E}\left(\bmu^T\mX\mv|\mU^T\mX\mV\right)\right]+\mathbb{E}\left[\text{Var}\left(\bmu^T\mX\mv|\mU^T\mX\mV\right)\right]\nonumber\\
    &\geq\text{Var}\left[\mathbb{E}\left(\bmu^T\mX\mv|\mU^T\mX\mV\right)\right]\nonumber\\
    &=\text{Var}\left[(\mU\eta_r)^T\mX(\mV\eta_c)\right],
\end{align}
where the last equality is from bilinearity of the conditional expectation.
Combining \eqref{eq:mean} and \eqref{eq:var} into the objective function in \eqref{eq:obj} yields,
\begin{align}\label{eq:2obj}
    L(\bmu,\mv,t)&\geq \text{Var}\left[(\mU\eta_r)^T\mX(\mV\eta_c)\right]+\lambda\mathbb{E}\left[1-\tilde Y\left[(\mU\eta_r)^T\mX(\mV\eta_c)-t\right]\right]_+\nonumber\\&\geq L(\mU\eta_r,\mV\eta_c,t).
\end{align}
Suppose $(\bmu,\mv)\notin S_{Y|\mX}$, then $\text{Var}\left(\bmu^T\mX\mv|\mU^T\mX\mV\right)>0$, which implies the strict inequality in \eqref{eq:var}. Thus, the inequality in\eqref{eq:2obj} is strict. This strict inequality in \eqref{eq:2obj} proves that  $(\bmu,\mv)$ cannot be the minimizer of $L(\bmu,\mv,t)$ unless $(\bmu,\mv)\in S_{Y|\mX}$.
\end{proof}

The bilinearity condition on $\mathbb{E}(\mX|\mU^T\mX\mV)$ is a generalization of the linearity condition  which is well-known and commonly assumed in the sufficient dimension reduction literature for the vector case \citep{li2009dimension,li2011principal,artemiou2016sufficient}.

Theorem \ref{thm:unbiased} implies that we can estimate the central subspace $S_{Y|\mX}$ by minimizing a series of the objective functions in Equation \eqref{eq:obj} with different $\tilde Y$s. We propose an estimation procedure of the PSMM at the sample level based on this intuition in the next section.

\section{Estimation procedure for the matrix sufficient dimension reduction}\label{sec:psmmalg}
We first introduce flip-flop algorithm to estimate mean and covariance matrices from i.i.d. sample $\{(\mX_i,Y_i)\}_{i=1}^n$. We then present the PSMM procedure to estimate the central subspace for matrices at the sample level.
\subsection{Flip-flop algorithm}
We assume that the feature matrix $\mX$ follows the matrix normal distribution, $\mX\sim \tM\tN_{d_1,d_2}(\mM, \mSigma_r,\mSigma_c)$, of which covariance matrix  has the form of,
\begin{align}\label{eq:Xcov}
    \text{Var}\left[\text{vec}(\mX)\right] = \mSigma_c\otimes \mSigma_r.
\end{align}
This covariance form \eqref{eq:Xcov} simplifies the objective function in \eqref{eq:obj} to
\begin{align}\label{eq:obj2}
L(\bmu,\mv,t) = (\bmu^T\mSigma_r\bmu)(\mv^T\mSigma_c\mv) + \lambda\mathbb{E}\left\{1-\tilde Y(\bmu^T(\mX-\mathbb{E}(\mX))\mv-t)\right\}_+.
\end{align}
In the sample level, we need to estimate mean and covariance matrices of the feature matrix $\mX$ for the objective function $L(\bmu,\mv,t)$ in \eqref{eq:obj2}. We propose a flip-flop algorithm for the estimation.

Let $\{(\mX_i,Y_i)\}_{i=1}^n$ be an i.i.d. sample of $(\mX,Y)\in\mathbb{R}^{d_1\times d_2}\times \mathbb{R}$. Given the sample matrices $\{\mX_i\}_{i=1}^n$, we estimate the mean and covariance matrices by
\begin{align}\label{eq:cov}
\bar \mX &= \frac{1}{n}\sum_{i=1}^n \mX_i\nonumber\\
    \hat\mSigma_r &= \frac{1}{d_2n}\sum_{i=1}^n (\mX_i-\bar \mX)\hat\mSigma_c^{-1}(\mX_i-\bar\mX)^T\nonumber\\
    \hat\mSigma_c &= \frac{1}{d_1n}\sum_{i=1}^n (\mX_i-\bar \mX)^T\hat\mSigma_r^{-1}(\mX_i-\bar\mX).
\end{align}
The covariance parameters does not have closed form unlike the mean parameter because two covariance matrices depend on each other. Thus, we compute their estimates iteratively until convergence based on \eqref{eq:cov}, which is known as “flip-flop” algorithm~\citep{dutilleul1999mle,glanz2018expectation}.
Notice that estimates for the mean and covariance matrices in \eqref{eq:cov} are the maximum likelihood estimator (MLE) when feature matrix follows the matrix normal distribution.  Details about MLE properties of the matrix normal distribution can be found in \cite{ros2016existence}.

There have been extensive studies about characteristics and statistical guarantees of the flip-flop algorithm.
The flip-flop algorithm is well known to converge to positive definite covariance matrices if and only if $n\geq \max(d_1/d_2,d_2/d_1)+1$~\citep{dutilleul1999mle}. More recently, \cite{franks2021near} provided the near-optimal sample complexity and established statistical guarantees of the flip-flop algorithm under the condition $n\geq C \frac{d_1}{d_1}\max\{\log d_2,\log^2 d_1\}$ where $C>0$ and $1<d_1\leq d_2$. In addition, they generalized all results to the higher-order tensor case. We leverage these results and use the outputs from the flip-flop algorithm to estimate the central subspace for matrices.

\subsection{The PSMM algorithm}\label{sec:procedure}
Now we present the PSMM algorithm for the matrix sufficient dimension reduction based on the observed samples. Suppose that the structural dimension $(r_1,r_2)$ of the central space $S_{Y|\mX}$ is known for now, i.e., $\text{dim}(S^{r}_{Y|\mX}) = r_1$ and $\text{dim}(S^{c}_{Y|\mX}) = r_2$. Unknown structural dimension case will be discussed in Section~\ref{sec:rank}. We summarize the estimation procedure as follows.
\begin{enumerate}
    \item[Step 1.] Calculate the sample mean $\bar \mX$ and covariance matrices $(\hat\mSigma_r,\hat\mSigma_c)$ using the flip-flop algorithm  in \eqref{eq:cov}.
    \item[Step 2.] Let $q_h$  be $(h/H)$-percentile of sample $\{Y_1,\ldots,Y_n\}$  for $h \in [H]$. Define $\tilde Y_i^h = \mathds{1}\{Y_i>q_h\}-\mathds{1}\{Y_i\leq q_h\}$ for each $h\in[H]$.
    
    \item[Step 3.] For each $h\in[H]$, find a solution $(\tilde \bmu^h, \tilde \mv^h,\tilde t^h)$ which minimizes
    \begin{align}\label{eq:sampleobj}
        (\bmu^T\hat\mSigma_r\bmu)(\mv^T\hat\mSigma_c\mv) + \frac{\lambda}{n}\sum_{i=1}^n\left\{1-\tilde Y_i^h(\bmu^T(\mX_i-\bar\mX)\mv-t)\right\}_+.
    \end{align}
    \item[Step 4.] Calculate the $r_1$ leading eigenvectors $(\hat \bmu_1,\hat \bmu_2,\ldots,\hat\bmu_{r_1})$ of $\hat \mU_n$ and $r_2$ leading eigenvectors $(\hat\mv_1,\hat\mv_2,\ldots,\hat\mv_{r_2})$ of $\hat \mV_n$, where we define
    \begin{align}\label{eq:construc}
        \hat\mU_n = \sum_{h = 1}^{H}\tilde\bmu^h(\tilde\bmu^h)^T\quad\text{and}\quad \hat\mV_n = \sum_{h = 1}^{H}\tilde\mv^h(\tilde\mv^h)^T.
    \end{align}
    \item[Step 5.] Estimate the central subspace $S_{Y|\mX}$ by 
    \begin{align}
        \hat S_{Y|\mX} = \text{span}(\{\hat\bmu_1,\ldots,\hat\bmu_{r_1}\})\times \text{span}(\{\hat\mv_1,\ldots,\hat\mv_{r_2}\}).
    \end{align}
\end{enumerate}
Step 3 optimizes the PSMM objective function at the sample level from \eqref{eq:obj2}. This objective function \eqref{eq:sampleobj} is  bi-convex  such that it is convex in $\bmu$ for fixed $\mv\in\mathbb{R}^{d_2}$ and convex in $\mv$ for fixed $\bmu\in\mathbb{R}^{d_1}$.
Thus, we minimize the equation \eqref{eq:sampleobj} using coordinate descent algorithm which solves convex optimization problem for one set of parameters holding the other fixed. We  update parameters $\bmu$ and $\mv$ iteratively based on Theorem~\ref{thm:update}.

Step 4 is to align components of column and row dimension reduction subspaces based on principal component analysis.

\begin{thm}\label{thm:update}
\begin{enumerate}
    \item If $\bmu^*$ minimizes \eqref{eq:sampleobj} over $\mathbb{R}^{d_1}$ for fixed $\mv\in\mathbb{R}^{d_2}$, then 
    \begin{align}
        \bmu^* = \frac{1}{2}\sum_{i=1}^n (\alpha_i^*\tilde Y_i^h)\frac{\hat\mSigma_r^{-1}(\mX_i-\bar\mX)\mv}{\mv^T\hat\mSigma_c\mv},
    \end{align}
    where $(\alpha^*_1,\ldots,\alpha^*_n)$ is the solution to the quadratic programming problem:
    \begin{align}
        \text{minimize }& -\sum_{i=1}^n\alpha_i +\frac{1}{4}\sum_{i=1}^n\sum_{j=1}^n\alpha_i\alpha_j\tilde Y_i^h\tilde Y_j^h\frac{\left((\mX_i-\bar\mX)\mv\right)^T\hat\mSigma_r^{-1}\left((\mX_i-\bar\mX)\mv\right)}{\mv^T\hat\mSigma_c\mv},\\&\text{subject to } \sum_{i=1}^n\alpha_i\tilde Y_i^h = 0 \text{ and } 0\leq \alpha_i\leq \frac{\lambda}{n} \text{ for all } i\in[n].
    \end{align}
    
    \item  If $\mv^*$ minimizes \eqref{eq:sampleobj} over $\mathbb{R}^{d_2}$ for fixed $\bmu\in\mathbb{R}^{d_1}$, then 
    \begin{align}
        \mv^* = \frac{1}{2}\sum_{i=1}^n (\beta_i^*\tilde Y_i^h)\frac{\hat\mSigma_c^{-1}(\mX_i-\bar\mX)^T\bmu}{\bmu^T\hat\mSigma_r\bmu},
    \end{align}
    where $(\beta^*_1,\ldots,\beta^*_n)$ is the solution to the quadratic programming problem:
    \begin{align}
        \text{minimize }& -\sum_{i=1}^n\beta_i +\frac{1}{4}\sum_{i=1}^n\sum_{j=1}^n\beta_i\beta_j\tilde Y_i^h\tilde Y_j^h\frac{\left((\mX_i-\bar\mX)^T\bmu\right)^T\hat\mSigma_c^{-1}\left((\mX_i-\bar\mX)^T\bmu\right)}{\bmu^T\hat\mSigma_r\bmu},\\&\text{subject to } \sum_{i=1}^n\beta_i\tilde Y_i^h = 0 \text{ and } 0\leq \beta_i\leq \frac{\lambda}{n} \text{ for all } i\in[n].
    \end{align}
\end{enumerate}
\end{thm}
\begin{proof}
Define $\tilde\bmu = \hat\mSigma_r^{1/2}\bmu$, $\tilde\mv = \hat\mSigma_c^{1/2}\mv$, and $\mZ_i = \hat\mSigma_r^{-1/2}(\mX_i-\bar\mX)\hat\mSigma_c^{-1/2}$. Then \eqref{eq:sampleobj} is equivalent to 
\begin{align}\label{eq:smmp}
    \|\tilde\bmu\|^2\|\tilde\mv\|^2 + \frac{\lambda}{n}\sum_{i=1}^n \left\{1- \tilde{Y}_i^h(\tilde\bmu^T\mZ_i\tilde\mv-t)\right\}_+.
\end{align}
Suppose that $\tilde\mv\in\mathbb{R}^{d_2}$ is fixed. Minimizing \eqref{eq:smmp} is then equivalent to 
\begin{align}
    \min_{\tilde \bmu, \mxi, t}& \|\tilde \bmu\|^2\|\tilde \mv\|^2 + \frac{\lambda}{n}\sum_{i=1}^n \xi_i,\\
    &\text{subject to } \tilde Y_i^h(\tilde \bmu^T \mZ_i\tilde \mv-t)\geq 1-\xi_i \text{ and } \xi_i\geq 0 \text { for all } i\in[n].
\end{align}
The Lagrange primal function is 
\begin{align}
    L_P = \|\tilde \bmu\|^2\|\tilde \mv\|^2 + \frac{\lambda}{n}\sum_{i=1}^n \xi_i-\sum_{i=1}^n\alpha_i\left[Y_i^h(\tilde \bmu^T \mZ_i\tilde \mv-t)-( 1-\xi_i)\right]-\sum_{i=1}^n \mu_i\xi_i,
\end{align}
which we minimize with respect to $\tilde \bmu, t,$ and $\xi_i$. Checking the first order condition yields 
\begin{align}\label{eq:formula}
    \tilde \bmu &= \sum_{i=1}^n \alpha_i \tilde Y_i^h \frac{\mZ_i \tilde \mv}{2\|\tilde\mv\|^2},\\
    0 &= \sum_{i=1}^n \alpha_i \tilde Y_i^h,\\
    \alpha_i &= \frac{\lambda}{n}-\mu_i \text{ for all } i\in[n],
\end{align}
with the positive constraints $\alpha_i,\mu_i, \xi_i\geq 0$ for all $i\in[n]$. By substituting the first order condition to the Lagrange primal function gives the dual objective function as 
\begin{align}
     \text{minimize }& -\sum_{i=1}^n \alpha_i + \frac{1}{4}\sum_{i=1}^n\sum_{j=1}^n \alpha_i \alpha_j \tilde Y_i^h \tilde Y_j^h \frac{(\mZ_i\tilde \mv)^T(\mZ_i\tilde \mv)}{\|\tilde\mv\|^2},\\&\text{ subject to } \sum_{i=1}^n\alpha_i\tilde Y_i^h = 0 \text{ and } 0\leq \alpha_i\leq \frac{\lambda}{n} \text{ for all } i\in[n].
\end{align}
Replacing back to original parameters in \eqref{eq:formula} with $\tilde\bmu = \hat\mSigma_r^{-1/2}\bmu$, $\tilde\mv = \hat\mSigma_c^{1/2}\mv$, and $\mZ_i = \hat\mSigma_r^{-1/2}(\mX_i-\bar\mX)\hat\mSigma_c^{-1/2}$ completes the first part for $\bmu^*.$ Updating $\mv$ for fixed $\bmu$ follows the same scheme so is omitted.
\end{proof}

\subsection{Determining the structural dimension}\label{sec:rank}
In practice, we need to estimate the structural dimension of $S^r_{Y|\mX}$ and $S^c_{Y|\mX}$. We propose to use a  modified  Bayesian information criterion (BIC) to estimate the unknown structural dimension $(r_1,r_2)$ in Step 4 in the previous section:
\begin{align}
    \text{BIC}(r_1) &= \sum_{i=1}^{r_1}\lambda_i(\hat\mU_n)-\lambda_1(\hat\mU_n)n^{-1/2}r_1,\\
    \text{BIC}(r_2) &= \sum_{i=1}^{r_2}\lambda_i(\hat\mV_n)-\lambda_1(\hat\mV_n)n^{-1/2}r_2,
\end{align}
where $\lambda_i(\mM)$ is the $i$-th largest eigenvector of a matrix $\mM$.  We choose the the structural dimension $(\hat r_1,\hat r_2)$ that minimizes the BIC.
Similar criteria have been used in \cite{zhu2006sliced,wang2008probability,li2011principal,artemiou2016sufficient}. The consistency for the estimated structural dimension is achieved when the Hessian matrix of \eqref{eq:obj2} is positive definite at an optimal point. Please see the details in Section~\ref{sec:disc}.

\section{Extension to higher order tensors}\label{sec:extend}
We extend the matrix sufficient dimension reduction to higher-order tensor case. Suppose that we have order-$K$ $(d_1,\ldots,d_K)$-dimensional feature tensors and responses $(\tX_i,Y_i)\in\mathbb{R}^{d_1\times \cdots\times d_K}\times \mathbb{R}$ for all $i \in[n].$ Our goal is to find $K$-number of matrices $\mU_k\in\mathbb{R}^{d_k\times r_k}$  $(r_k<d_k)$ for $k = 1,\ldots,K$ such that 
\begin{align}\label{eq:hmain}
    Y \indep \tX| \tX\times_1\mU_1\times_2\cdots\times_K\mU_K.
\end{align}
Then we keep all information of feature tensor $\tX\in\mathbb{R}^{d_1\times \cdots\times d_K}$ to predict the response $Y$ only with the reduced feature dimension $\tX\times_1\mU_1\times_2\cdots\times_K\mU_K\in\mathbb{R}^{r_1\times \cdots\times r_K}.$
Notice that  equation \eqref{eq:hmain}  is reduced to the matrix sufficient dimension reduction problem \eqref{eq:main} in the matrix case ($K = 2$). 
Similar to the matrix case, we define the central mode-$k$ subspace, denoted by $\tS^k_{Y|\tX}$, as
\begin{align}
    \tS^k_{Y|\tX} =\underset{\{\mU_k\colon Y\indep \tX|\tX\times_1\mU_1\times_2\cdots\times_K\mU_K\}}{\cap}  \text{span}(\mU_k)
\end{align}
The subspace $S_{Y|\tX} = \bigtimes_{k=1}^K S^k_{Y|\tX}$ is called the central subspace for higher-order tensor.

We propose a principal support tensor machine (PSTM) generalizing the PSMM to the higher-order tensor case.
We consider the following objective function of the PSTM at the population level.
\begin{align}\label{eq:objh}
    L(\{\bmu_k\}_{k=1}^K,t) &= \text{Var}(\tX\times_1\bmu_1\times_2\cdots\times_K\bmu_K)\\&+\lambda \mathbb{E}\left\{1-\tilde Y\left((\tX-\mathbb{E}(\tX))\times_1\bmu_1\times_2\cdots\times_K\bmu_K-t\right)\right\}_+,
\end{align}
where we define random variable $\tilde Y = \mathds{1}\{Y\in A_1)-\mathds{1}\{Y\in A_2\}$ for arbitrary disjoint subsets  $A_1$ and $A_2$ of the support of $Y$.

Using the similar proof argument in Theorem~\ref{thm:unbiased}, we can prove the following theorem.
\begin{thm}\label{thm:unbiasedh}
Suppose that $\mathbb{E}(\tX|\tX\times_1\mU_1\times_2\cdots\times_K\mU_K)$ is a multilinear function of $\tX\times_1\mU_1\times_2\cdots\times_K\mU_K$, where  $\{\mU_k\}_{k=1}^K$ are matrices as defined in \eqref{eq:hmain}. If $(\bmu_1^*,\ldots,\bmu_K^*,t^*)$ minimizes the objective function \eqref{eq:objh} among all $(\bmu_1,\ldots,\bmu_K,t)\in\mathbb{R}^{d_1}\times \cdots\times \mathbb{R}^{d_K}\times\mathbb{R}$, then $(\bmu_1^*,\ldots,\bmu_K^*)\in S_{Y|\tX}.$
\end{thm}
Theorem~\ref{thm:unbiasedh} provides the guidance for estimating the central subspace for higher-order tensors. We estimate the central subspace by minimizing a series of objective function of the PSTM and aggregating all minimizers. Since the estimation procedure is very similar to the matrix case, we only highlight the major differences here.

In the sample level estimation, the objective function of the PSTM in Step 3 in Section~\ref{sec:psmmalg} becomes:
\begin{align}\label{eq:sampleobjh}
        \prod_{k=1}^K(\bmu_k^T\hat\mSigma_k\bmu_k) + \frac{\lambda}{n}\sum_{i=1}^n\left\{1-\tilde Y_i^h\left((\tX_i-\bar\tX)\times_1\bmu_1\times_2\cdots\times_K\bmu_K-t\right)\right\}_+,
\end{align}
where $\hat\mSigma_k$ is obtained from the flip-flop algorithm based on the tensor normal model, whose covariance has the structure $\text{Var}\left(\text{vec}(\tX)\right) = \mSigma_1\otimes\cdots\otimes \mSigma_K$. We skip the details of the flip-flop algorithm here, but note that the algorithm and its consistency for the higher-order tensor case can be found in Section 2 of \cite{franks2021near}.
To minimize the objective function \eqref{eq:sampleobjh}, we leverage support tensor machine (STM) algorithms. To be specific, let $\bmu'_k = \hat\mSigma_k^{1/2}\bmu_k$ and $\tZ_i = (\tX_i-\bar\tX)\times_1\hat\mSigma_1^{-1/2}\times_2\cdots\times_K\hat\mSigma_K^{-1/2}.$ Then, we rewrite \eqref{eq:sampleobjh} as:
\begin{align}\label{eq:sampleobjh2}
        \prod_{k=1}^K\|\bmu_k'^T\|^2 + \frac{\lambda}{n}\sum_{i=1}^n\left\{1-\tilde Y_i^h\left(\tZ_i\times_1\bmu_1'\times_2\cdots\times_K\bmu_K'-t\right)\right\}_+,
\end{align}
which is the objective function of regular STM with a rank-1 constraint introduced in \cite{kotsia2011support,kotsia2012higher}. Thus we can apply standard STM algorithm to solve the optimization problem. 
Finally, we aggregate the optimizers $(\tilde\bmu_1^h,\ldots,\tilde\bmu_K^h)$ in \eqref{eq:sampleobjh} for $h\in [H]$ and estimate the central subspace $\tS_{Y|\tX}$ as in Step 4-5 in Section~\ref{sec:psmmalg}.

\section{Numerical analysis}\label{sec:sim}
In this section, we analyze the synthetic and real world datasets to demonstrate the performance of our PSMM.
\subsection{Synthetic data}
We compare the performance of the PSMM with existing matrix sufficient dimension reduction methods. Our comparison involves two aspects.  Firstly, we compare the performance of matrix-based sufficient dimension reduction methods with a conventional vector-based method. Secondly, we compare the performance of our method with existing sufficient dimension reduction methods for matrix predictors.
\begin{itemize}
    \item Principal Support Vector Machine (PSVM)~\citep{li2011principal,artemiou2016sufficient} uses the support vector machine for sufficient dimension reduction. We vectorize feature matrices and apply this vector-based sufficient dimension reduction method.
    \item Folded Sliced Inverse Regression (folded-SIR)~\citep{li2010dimension,ding2015tensor} is based on the sliced inverse regression method proposed in \cite{li1991sliced}. Folded-SIR generalizes the vector-based method to the matrix case and involves the first order inverse moment.
    \item Folded Directional Regression (folded-DR)~\citep{li2010dimension} is based on the directional regression method  proposed in \citep{li2007directional}. Folded-DR  generalizes the vector-based method to the matrix case and involves the second order inverse moment.
\end{itemize}

We use the following models:  
\begin{align}
    &\text{Model 1}:\quad Y = \exp({\mX_{11}})+\mX_{12} + \epsilon,\\
    &\text{Model 2}:\quad Y = {\mX_{11}}/\{0.5+(\mX_{12}+1)^2\} + \epsilon,\\
    &\text{Model 3}:\quad Y = \mX_{11}(\mX_{12}+\mX_{21}+1) + \mX_{11} + \epsilon.
\end{align}
where the feature matrix $\mX\in\mathbb{R}^{d\times d}$ is i.i.d. drawn from $\tM\tN_{d,d}(0_{d\times d},\mI_d,\mI_d)$ and the noise $\epsilon$ is i.i.d. drawn from $N(0,0.2^2).$ The central subspace of feature matrix $\mX$ is $\text{span}(\{e_1\})\times \text{span}(\{e_1,e_2\})$ in Model 1 and 2, while  $\text{span}(\{e_1,e_2\})\times \text{span}(\{e_1,e_2\})$ in Model 3. We vary the sample size $n\in\{100,200,\ldots,500\}$ and matrix dimension $d\in\{5,10\}.$
We choose $q_h$ in our algorithm to be $(h/H)$-percentile of sample $\{Y_1,\ldots,Y_n\}$  for $h \in [H]$. We set the hyperparameter $H = 10$ and $\lambda = 100.$ 

We use the distance measure suggested by \cite{Li2005ContourRA,li2010dimension} to evaluate the performance of each model. Specifically, let $\tS^r$ and $\tS^c$ be the true central row and column subspace respectively while  $\hat\tS^r$ and $\hat\tS^c$ be the estimated one. Then we define the estimation error by
\begin{align}
    \text{dist}(\tS^r\otimes\tS^c, \hat\tS^r\otimes\hat\tS^c)= \left\|\mP_{\tS^r\otimes \tS^c}-\mP_{\hat\tS^r\otimes \hat\tS^c}\right\|_F,
\end{align}
where $\mP_{\tS}$ is an orthogonal projection on to the subspace $\tS$ and $\|\cdot\|_F$ is the matrix Frobenius norm. All summary statistics are averaged across 20 replicates.

Figure~\ref{fig:sim} illustrates the estimation error of various sufficient dimension reduction methods for models 1-3, evaluated across different sample sizes and feature matrix dimensions. Notably, the PSMM algorithm consistently outperforms the other methods in all scenarios.
We verified a clear advantage of matrix-based methods over the vector-based method, as all matrix-based methods outperformed the PSVM in all scenarios. In addition, the PSMM showed better performance than alternative matrix-based methods. The intuition behind this improvement can be similarly explained as in \cite{li2011principal}. Since SIR and DR methods tend to downweight the slice means near the center of the data points, they are not suitable for cases where the regression function is more accurately estimated near the center of the data points, which is often true in many cases \cite{kutner2004applied}. By contrast, the PSMM uses the hyperplane that separates the datapoints, so it does not have a downweighting effect on the data near the center. Finally, we see that all algorithms show a polynomial decaying pattern as the sample size increases, which implies the consistency of estimators. We also found that the performance of algorithms tends to decrease when the matrix dimension increases. This is not surprising because the larger matrix dimension implies a bigger space to search.

\begin{figure}
\centering
\subfigure[Estimation error for Model 1]{\includegraphics[width=1\linewidth]{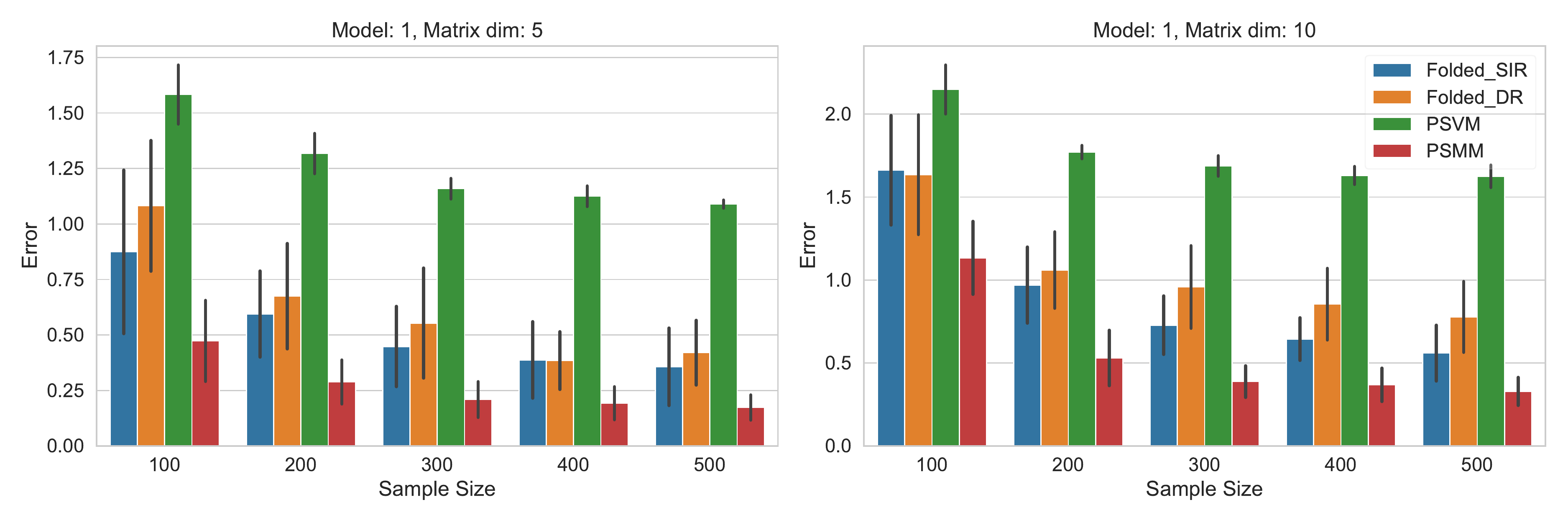}}
\subfigure[Estimation error for Model 2]{\includegraphics[width=1\linewidth]{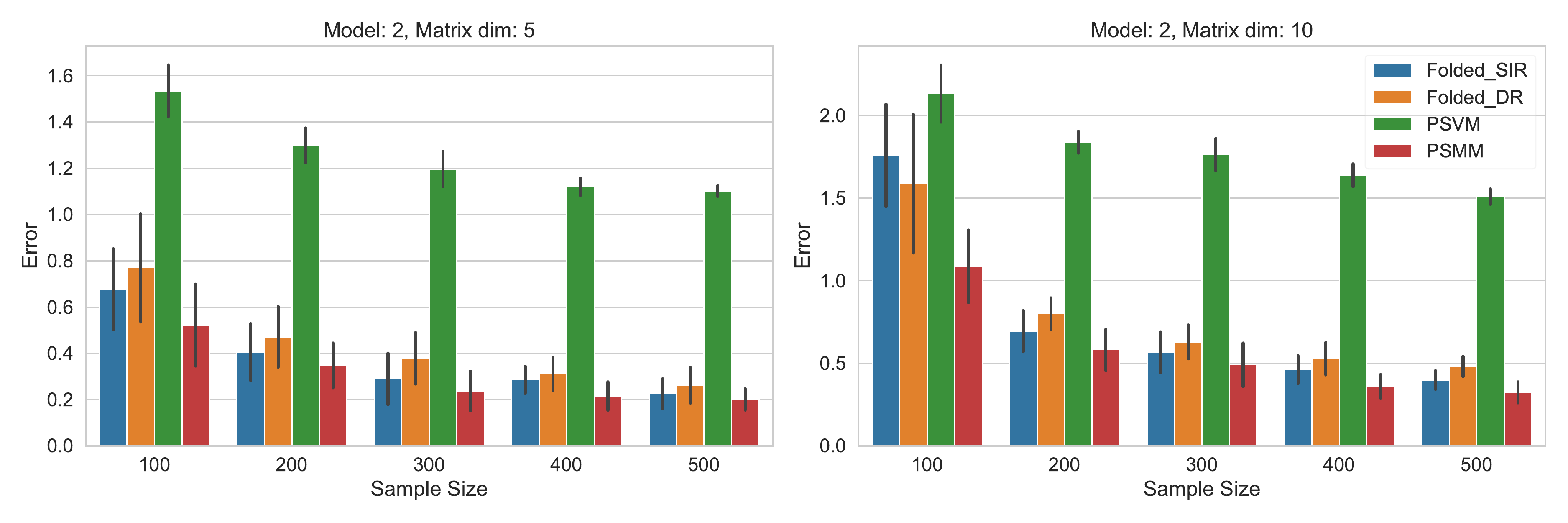}}
\subfigure[Estimation error for Model 3]{\includegraphics[width=1\linewidth]{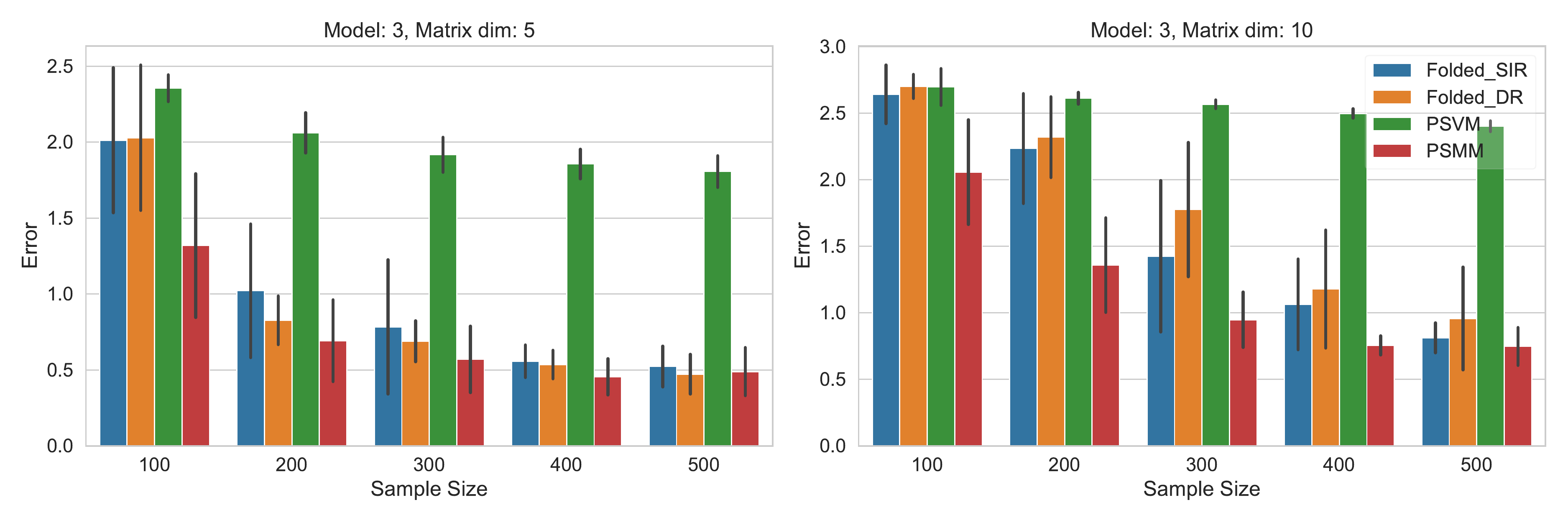}}
\caption{The estimation error of four methods across different sample size and feature matrix dimension.}
\label{fig:sim}
\end{figure}

\subsection{Application to MRN-114 human brain connectivity data}
We apply our PSMM to MRN-114 dataset. This dataset consists of the structural connectivity of the 68 brain nodes along with their cognitive ability measured by FSIQ (Full Scale Intelligence Quotient)  score for a total of 114 subjects~\citep{jung2007parieto,wang2017bayesian}. We convert the connectivity data into adjacency matrices $\mX_i \in \mathbb{R}^{68\times 68}$ for $i \in [114]$, where each entry indicates the presence or absence of fiber connections between 68 distinct brain regions.  The corresponding response to each adjacency matrix $\mX_i$ is the FSIQ score $Y_i$, which ranges from 86 to 144.  We apply the PSMM algorithm with the input $\{(\mX_i,Y_i)\}_{i=1}^{114}$ and estimate the matrix $\mU\in\mathbb{R}^{64\times r}$ such that $Y\indep \mX|\mU^T\mX\mU$. We set $r = 2$ for ease of visualization and interpretation.
Based on the estimated matrix $\hat\mU = (\hat \bmu_1, \hat \bmu_2)\in\mathbb{R}^{64\times 2}$, we calculate the reduced feature variables $V^1_i$, $V^2_i$, and $V^3_i$ for each observation $i \in [114]$, defined as follows:
\begin{align}
    V^1_i &= \hat\bmu_1^T\mX_i\hat\bmu_1 = \langle \hat\bmu_1\hat\bmu_1^T, \mX_i \rangle,\\
    V^2_i &= \hat\bmu_2^T\mX_i\hat\bmu_2 = \langle \hat\bmu_2\hat\bmu_2^T, \mX_i \rangle,\\
    V^3_i &= \hat\bmu_1^T\mX_i\hat\bmu_2 = \langle \hat\bmu_1\hat\bmu_2^T, \mX_i \rangle.
\end{align}
Figure~\ref{fig:loading} visualizes the loading matrices for the reduced feature variables $V^1$, $V^2$, and $V^3$ in order. We find that all loading matrices have a sparse structure, which implies that only some brain networks significantly explain the FSIQ score. Figure~\ref{fig:brain}(a) plots the reduced feature variables $V^1,V^2,$ and $V^3$ along with the FSIQ scores of individuals. Surprisingly, the three feature variables capture the trend of FSIQ very well. For example, individuals who have a large negative value for $V^1$, a small negative value for $V^2$, and a positive value for $V^3$ tend to have higher FSIQ scores, while those who have a small negative value for $V^1$, a large negative value for $V^2$, and a negative value for $V^3$ are inclined to have lower FSIQ scores. Furthermore, we inspected the entries of the loading matrix for $V^1$ (i.e., $\hat\bmu_1\hat\bmu_1^t$) and plotted the brain connections for the top 10 negative value, as shown in Figure~\ref{fig:brain}(b). Interestingly, a brain node called the right isthmuscingulate had multiple edges with other nodes. It is well-known that this region is involved in various cognitive and emotional processes and has been found to be associated with certain aspects of cognitive function, including intelligence~\citep{vogt2006cytology,li2014association}. In addition, we observed that the connections were mostly inter-hemispheric, excluding the connection with the right isthmuscingulate. This finding is also in agreement with recent studies on the correlation between brain connectivity and intelligence~\citep{wang2017bayesian,lee2021beyond}.

\begin{figure}
    \centering
    \includegraphics[width = \textwidth]{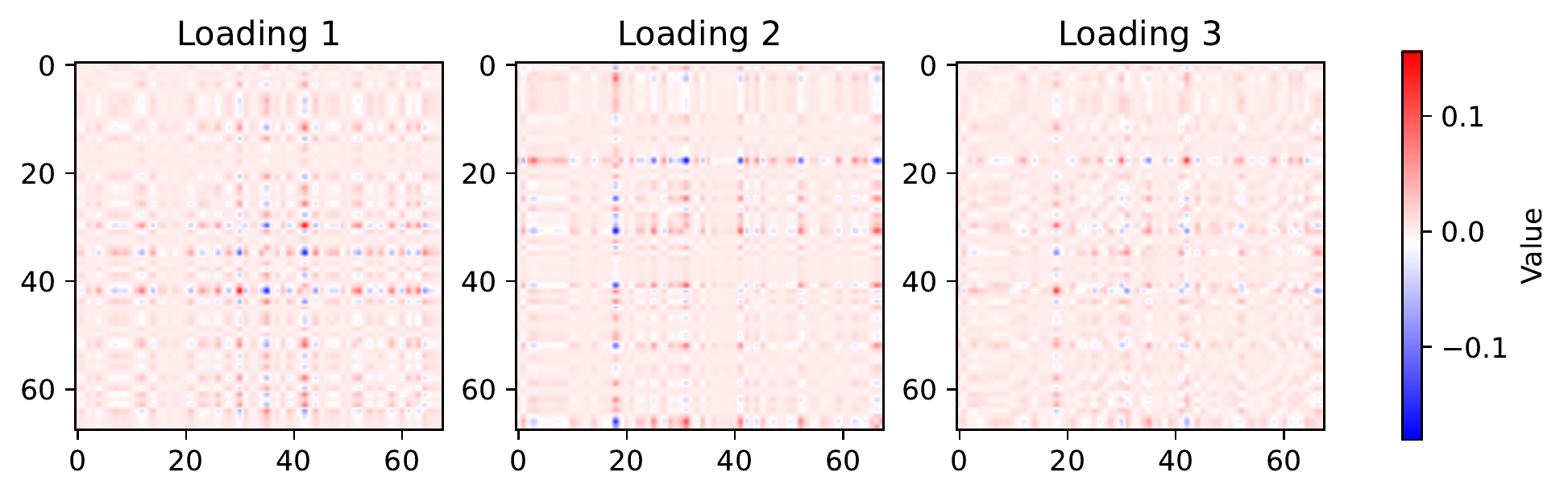}
    \caption{Loading matrices for the reduced feature variables $V^1,V^1$ and $V^3$ in order.}
    \label{fig:loading}
\end{figure}

\begin{figure}
    \centering
    \subfigure[]{\includegraphics[width = 0.49\textwidth]{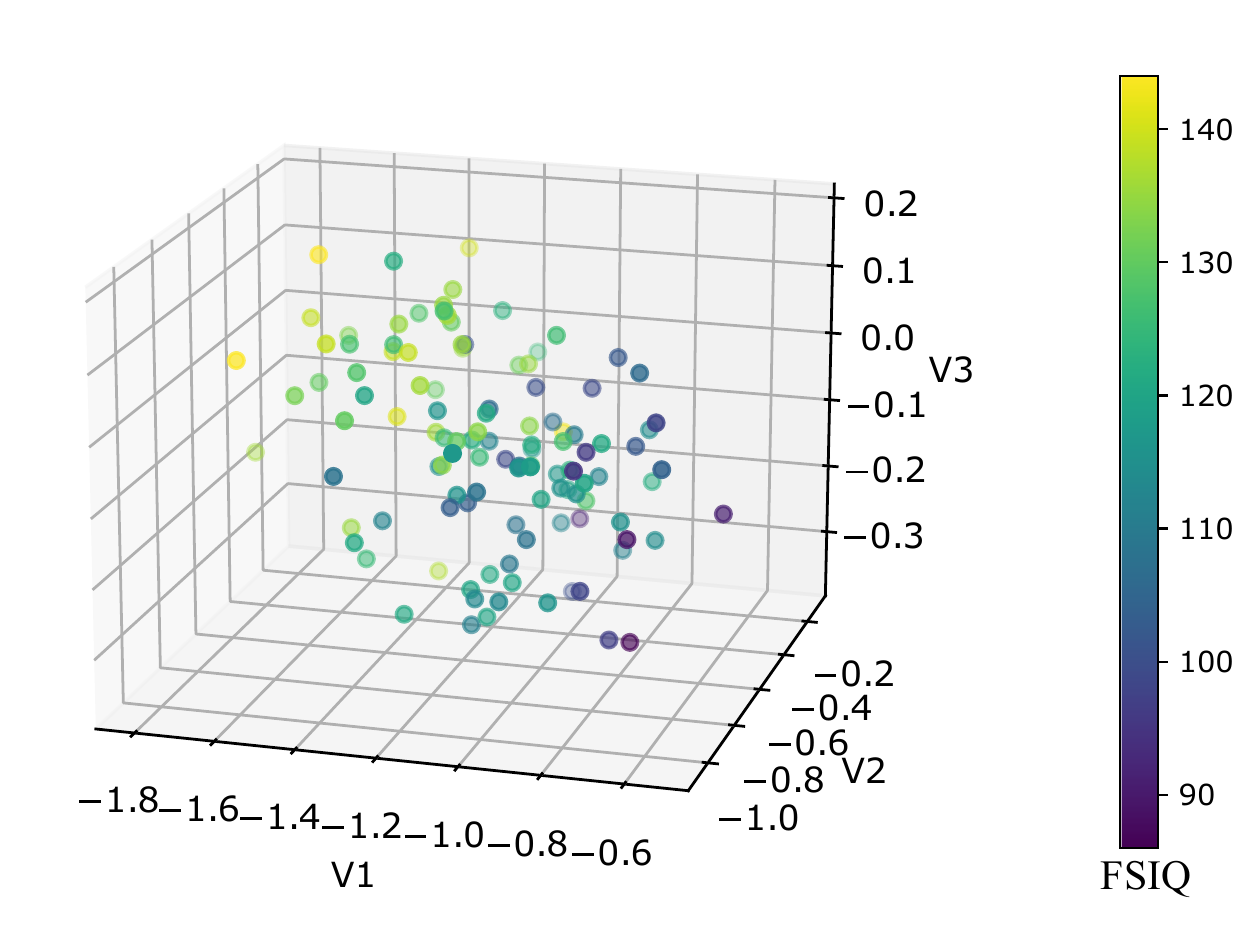}}
    \subfigure[]{\includegraphics[width = 0.49\textwidth]{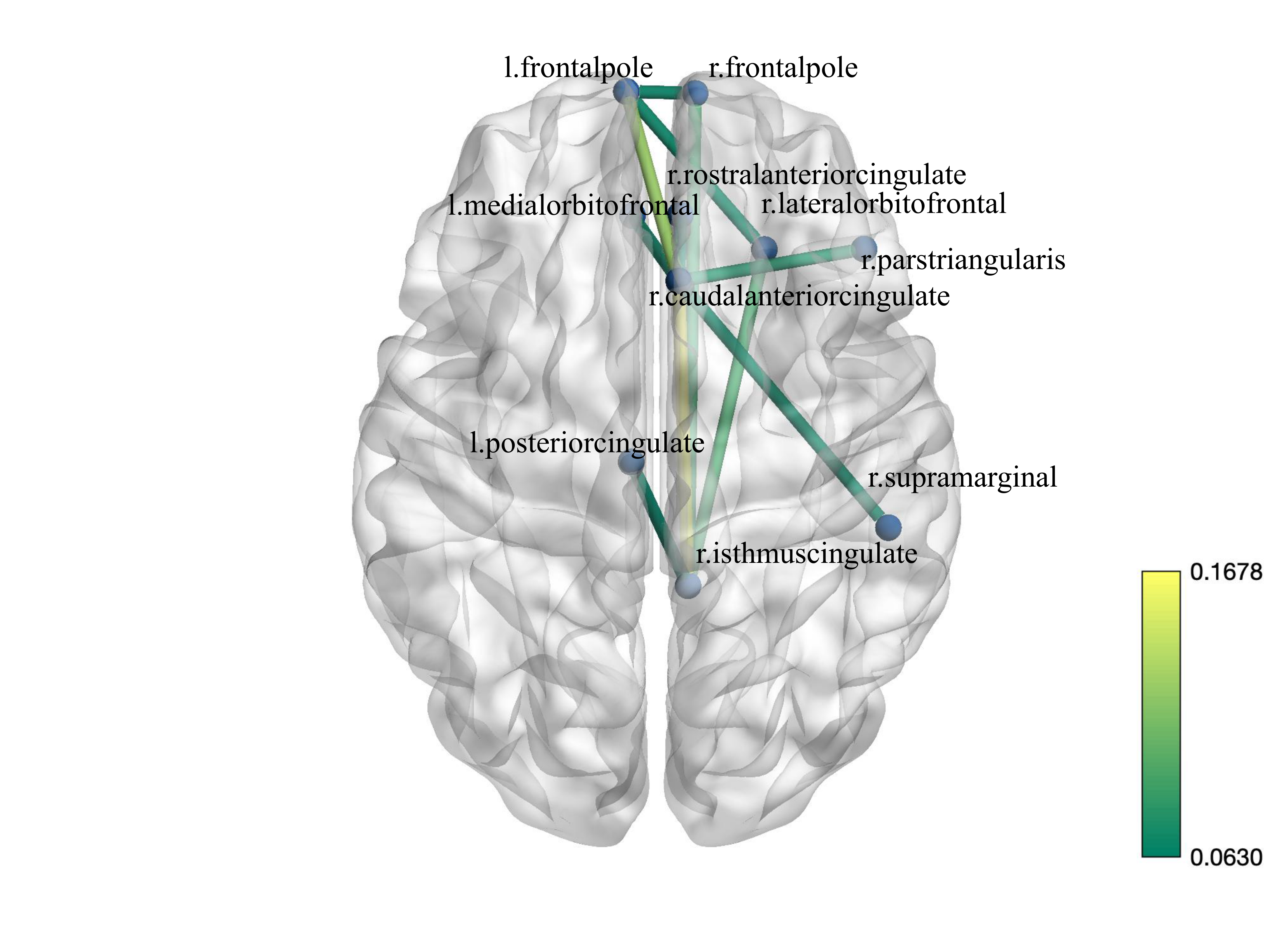}}
    \caption{(a) Scatterplot of three reduced feature variables estimated by the PSMM. The color of points shows the corresponding FSIQ scores of individuals (b) Top 10 FSIQ-associated edges in brain connectivity data.}
    \label{fig:brain}
\end{figure}

\section{Conclusion and discussion}\label{sec:disc}
We propose a new matrix sufficient dimension reduction method called the Principal Support Matrix Machine (PSMM). The PSMM preserves the matrix structure of predictors and enjoys more accurate estimation of the central subspace compared to other existing dimension reduction methods. Numerical analysis demonstrates the effectiveness and applicability of our PSMM.

There are several possible extensions from our work. Although we observe the empirical evidence that our estimation error converges with polynomial decays, we have not shown statistical convergence of the estimator.  In fact, we can leverage the asymptotic results of SVM in \cite{jiang2008estimating,koo2008bahadur} to construct the consistency. Suppose that the Hessian matrix of the objective function $L(\bmu,\mv,t)$ in \eqref{eq:obj2} is positive definite at an optimal point $(\bmu^*,\mv^*,t^*)$. Then, combining similar proof argument of Theorem 2 in \cite{jiang2008estimating} and construction of $(\hat \mU_n,\hat\mV_n)$ in  \eqref{eq:construc} yields that 
\begin{align}\label{eq:consist}
    \hat\mU_n-\mU = \tO_p(n^{-1/2}) \quad\text{and}\quad \hat\mV_n-\mV = \tO_p(n^{-1/2}).
\end{align}
Therefore, we achieve the consistency of the PSMM by \eqref{eq:consist} and \cite{bura2008distribution} such that $\hat\mU-\mU = \hat\mV-\mV =\tO_p(n^{-1/2})$, 
where $\hat\mU$ and $\hat\mV$ are outputs from the PSMM. 
Unlike vector case, however, positive definiteness of the Hessian matrix is not guaranteed. To be specific, we show that the Hessian of $L(\bmu,\mv,t)$  has the explicit form under some technical conditions as
\begin{align}
     \mH = \mH_1+\lambda\sum_{\tilde y = -1,1}(\mH_2+\mH_3)\mathbb{P}\left[\tilde Y = \tilde y\right],
 \end{align}
 where we define
\begin{align}\label{eq:Hessian}
    \mH_1 &= 2\begin{pmatrix}(\mv^T\mSigma_c\mv)\mSigma_r & 2 \mSigma_r \bmu\mv^T\mSigma_c & 0_{d_1\times 1}\\ 
    2\mSigma_c \mv\bmu^T\mSigma_r & (\bmu^T\mSigma_r\bmu)\mSigma_c &0_{d_2\times 1}\\ 0_{1\times d_1} & 0_{1\times d_2} & 0 \end{pmatrix}, \\ 
    \mH_2 &=  \mathbb{E}\left[\begin{pmatrix}\mX\mv \\\mX^T\bmu\\-1\end{pmatrix}\begin{pmatrix}\mX\mv \\\mX^T\bmu\\-1\end{pmatrix}^T\middle\vert \bmu^T\mX\mv = t+\tilde y\right]f_{\bmu^T\mX\mv|\tilde Y}(t+\tilde y),\\
  \mH_3 & = -\mathbb{E}\left[\begin{pmatrix}0 &\mX& 0\\ \mX^T&0&0\\0&0&0\end{pmatrix}\mathds{1}\{\tilde y(\bmu^T\mX\mv-t)<1\}\right].
\end{align}
Here $f_{\cdot|\cdot}$ denotes the conditional probability density function. Notice that the positive definite Hessian comes free in the SVM by its convexity. However, checking the positive definiteness of the Hessian is not trivial for the SMM due to its non-convexity. Finding an explicit condition for the Hessian matrix to be positive definite at an optimal point warrants future research.

 Constructing the consistency of the BIC is another interesting question. In Section~\ref{sec:rank}, we propose the modified BIC to estimate the true structural dimension $(r_1,r_2)$ of the central subspace. We can show the consistency of the BIC under the assumption that the Hessian in \eqref{eq:2obj} is positive definite at an optimal point. We briefly sketch the proof here. As mentioned above, the positive definite Hessian guarantees the consistency of estimator for the central subspace by Equation \eqref{eq:consist}. Under the this consistency,  we can follow the same proof argument in Theorem 8 in \cite{li2011principal}. Finally, we set constants in Theorem 8 in \cite{li2011principal} as $c_1(n) = n^{1/2}\log n$ and $c_2(k) = k$, which completes the proof for the consistency of the BIC, $\lim_{n\rightarrow\infty}\mathbb{P}(\hat r_1 = r_1)$ and $\lim_{n\rightarrow\infty}\mathbb{P}(\hat r_2 = r_2)$.


 \bibliographystyle{chicago}      
 \bibliography{stm}   

\end{document}